\documentclass[letterpaper, 10 pt, conference]{ieeeconf}  

\IEEEoverridecommandlockouts                              

\usepackage{times} 
\usepackage{amsmath} 
\usepackage{amssymb}  
\usepackage{algorithmic}
\usepackage{algorithm}

\newcommand{\B}{{\cal B}}
\newcommand{\J}{{\cal J}}

\newcommand{\supp}{\mbox{supp }}
\newcommand{\spark}{\mbox{spark }}

\newcommand{\sbold}{\mathbf{s}}

\newcommand{\ybold}{\mathbf{y}}

\newcommand{\etabold}{\mathbf{\boldsymbol{\eta}}}
\newcommand{\rbold}{\mathbf{r}}
\newcommand{\beq}{\begin{equation}}
\newcommand{\eeq}{\end{equation}}
\newcommand{\bmat}{\left[ \begin{array}}
\newcommand{\emat}{\end{array} \right]}

\newcommand{\param}{\xi}

\newtheorem{definition}{Definition}[section]

\newtheorem{thm}{Theorem}[section]
\newtheorem{lem}[thm]{Lemma}
\newtheorem{cor}[thm]{Corollary}
\newtheorem{example}[thm]{Example}

\title{\LARGE \bf
Vulnerability of linear systems against sensor attacks--a system's security index
}

\author{Michelle S. Chong$^{1}$ and Margreta Kuijper$^{2}$
\thanks{$^{1}$Michelle S. Chong is with the Department of Automatic Control, Lund University, SE-221 00 Lund, Sweden.
        {\tt\small michelle.chong@control.lth.se}}%
\thanks{$^{2}$Margreta Kuijper is with the Department of Electrical and Electronic Engineering, University of Melbourne, Australia.
        {\tt\small mkuijper@unimelb.edu.au}}%
}

\begin{document}

\maketitle
\thispagestyle{empty}
\pagestyle{empty}

\begin{abstract}
The `security index' of a discrete-time LTI system under sensor attacks is introduced as a quantitative measure on the security of an observable system. We derive ideas from error control coding theory to provide sufficient conditions for attack detection and correction. 
\end{abstract}

\section{Introduction}
The security of control systems against adversarial attacks is a challenge to maintain when the adversary knows the workings of any component of the system and has garnered access, with the malicious intent of causing disruption. This has lead to a proliferation of works in tackling this issue, in particular in detecting the occurrence of an attack \cite{pasqualetti2013attack, pasqualetti2012attack, pasqualetti2015divide}, or in designing resilient control or estimation algorithms, see \cite{fawziTD14, shoukry2013event, ChongWakaikiHespanhaACC15, chenKarMouraICASSP15, sandbergTJ2010} and many more.

In this paper, we concentrate on LTI systems where the sensing component has been compromised by the attacker, who has full knowledge of the system. The vulnerability of the sensors is modelled by an additive attack signal to the sensor measurements, which is non-zero when the particular sensor is compromised. Inspired by ideas in coding theory, we introduce the notion of the `security index' for linear systems, a quantitative measure of the vulnerability of a system to sensor attacks. While ideas from error control coding have already been employed to this context in recent literature \cite{fawziTD14}, our aim is to further strengthen this link. Our notion of a `security index' is formulated based on the measurement time series from all sensors and is analogous to the notion of the `minimum distance' of a code in error control coding theory. We demonstrate that by using ideas from coding theory, the formulation simplifies the approach in \cite{fawziTD14}, leading to new results. Particularly, we express the `security index' of a system in terms of different representations of the system concerned. 
 
Previous works in state estimation for systems under sensor attacks include \cite{ChongWakaikiHespanhaACC15, fawziTD14, chenKarMouraICASSP15, pasqualetti2013attack, pasqualetti2015divide, pajic2015attack}. There is a consensus with \cite{fawziTD14} and \cite{ChongWakaikiHespanhaACC15} that the states of an LTI system can only be reconstructed if strictly less than half of the sensors are under attack. We will see in this paper that this condition is also derived when approached with ideas from coding theory. Other related works are \cite{HendrickxJJSS14, sandbergTJ2010} which focus on power networks. It is in this specialised setting that the authors of \cite{sandbergTJ2010} introduce the terminology `security index', which we adopted for a broader context in this paper. The presence of measurement noise has been considered in \cite{pajic2015attack}, which we do not consider, but is the subject of further work. 
\newline
\noindent{\em Notation}: We denote the set of integers and complex numbers as $\mathbb{Z}$ and $\mathbb{C}$, respectively. The notation $\mathbb{Z}_{+}$ is used to denote the set of positive integers including $0$.

\section{Problem formulation}
We consider a discrete-time, observable linear time-invariant (LTI) system $\Sigma$ given by a $n\times n$ state matrix $A$ and a $N\times n$ observation matrix $C$, defined as follows:
\begin{eqnarray}
x(t+1) & = &  Ax(t) \label{eq_Asystem}\\
y(t) & = &  Cx(t). \label{eq_Csystem}
\end{eqnarray}
The {\em behavior} $\B$ of the system is defined as the set of all possible output trajectories $\ybold:\mathbb{Z}_+ \mapsto \mathbb{C}^N$ that satisfy the system's equations for some initial condition $x(0) \in \mathbb{C}^n$. Due to the time-invariant finite dimensional nature of the underlying system, the behavior $\B$ has the following two properties:
\begin{itemize}
\item $\B$ is left-shift invariant: if $\ybold \in \B$ then $\sigma \ybold \in \B$, where the shift operator $\sigma$ is defined via $\sigma y(t) := y(t+1)$.
\item $\B$ is autonomous: there exists $T \in \mathbb{Z}_+$ such that for any $\ybold \in \B$ and $\tilde \ybold \in \B$ we have that $\ybold|_{[0,T] }=\tilde \ybold|_{[0,T] }$ implies that $\ybold = \tilde \ybold$.
\end{itemize} 
We assess the vulnerability of an LTI system via its measurable outputs, which may have been compromised by an attacker. While the usual assumption for many applications is that the matrix $C$ has full row rank, we do not make such an assumption in our setting. This setting occurs in the case where each sensor is measuring a local part of the system, such as in sensor networks implemented in a large geographical location. To aid in the introduction of a measure of a system's security against sensor attacks, we define the following for a system's trajectory.

\begin{definition} \label{def:supp_y_trajectory} The {\em support} of a trajectory $\ybold:\mathbb{Z}_+ \mapsto \mathbb{C}^N$, denoted by $\supp(\ybold)$, is defined as the set of indices $i$ in $\{ 0, 1, 2, \ldots , N\}$ such that its component trajectory $\ybold_i:\mathbb{Z}_+ \mapsto \mathbb{C}$ is not the zero trajectory. 
\end{definition}
\begin{definition} The {\em weight} of a trajectory $\ybold$, denoted by $\|\ybold\|$, is defined as $|\supp(\ybold)|$, i.e., the number of components of $\ybold$ that are not the zero trajectory.
\end{definition}

We now introduce a concept that is central to this paper:
\begin{definition}\label{def_security} The {\em security index} of the system $\Sigma$ is defined as 
\[
\delta (\Sigma ) := \min _{0\neq \ybold \in \B} \|\ybold\| .
\]
\end{definition}
This notion plays a paramount role in our investigation into the resilience of the system under adversarial attack. More precisely, we consider attacks on the system $\Sigma$ that result in the system $\Sigma'$ given by:
\begin{eqnarray*}
\Sigma': \; x(t+1) & = &  Ax(t) \\
r(t) & = &  Cx(t) + \eta (t),
\end{eqnarray*}
where $\etabold:\mathbb{Z}_+ \mapsto \mathbb{C}^N$ is the unknown attack signal and $\rbold:\mathbb{Z}_+ \mapsto \mathbb{C}^N$ is the known received signal. Thus we are focusing exclusively on scenarios where the system's outputs (= sensors) are attacked. The behavior $\B'$ of the system $\Sigma'$ is defined as the set of all possible trajectories $\rbold$ that satisfy the above equations for some initial condition $x(0)$ and some attack signal $\etabold$. We consider the following two problems:

\begin{quote}
{\em Problem 1 (attack detection):} Given received signal $\rbold \in \B'$, detect that $\rbold \notin \B$. 
\end{quote}
\begin{quote}
{\em Problem 2 (attack correction):} Given received signal $\rbold\in \B'$, find $\ybold \in \B$ such that $\|\rbold - \ybold\| $ is minimal. 
\end{quote}
In the sections that follow, we derive conditions such that the problems above are solvable in a tractable manner.

\section{Conditions for attack detection and correction using a system's security index}
The first question that arises is: under which conditions on the attack signal $\etabold$ are these problems solvable? We have the following results. 

\begin{thm}\label{thm_detect}
Suppose the received signal $\rbold \in \B'$ corresponds to an attack $\etabold$ with $\|\etabold\| $ an unknown non-zero value $< \delta (\Sigma  )$. Then $\rbold \notin \B$, i.e., attack detection is possible. 
\end{thm}
\begin{proof}
Let $\rbold$ and $\etabold$ be as stated in the theorem and let $\ybold \in \B$ be such that $\rbold = \ybold + \etabold$. Then $\rbold - \ybold \notin \B$ because of Definition~\ref{def_security} and the assumption that $0 \neq \|\etabold\| < \delta (\Sigma  )$. Since $\B$ is linear it then follows that $\rbold \notin \B$.	
\end{proof}

Consequently, we can interpret the security index $\delta (\Sigma  )$ as the minimum number of sensors that an attacker needs to compromise without being detected. We call the system $\Sigma $ {\em maximally secure} if $\delta (\Sigma  ) = N$. It is easily seen that generically, systems described by equations~\eqref{eq_Asystem}-\eqref{eq_Csystem} are maximally secure.

\begin{thm}\label{thm_correct}
Suppose the received signal $\rbold \in \B'$ corresponds to an attack $\etabold$ with $\|\etabold\| $ an unknown value $< \delta (\Sigma  )/2$. Then there exists a unique $\ybold \in \B$ such that $\|\rbold - \ybold\| $ is minimal, i.e. unique attack correction is possible.   
\end{thm}
\begin{proof}
		Let $\rbold$ be as stated in the theorem and let $\ybold_1 \in \B$ and $\etabold_1$ be such that $\rbold = \ybold_1 + \etabold_1$ with $\|\etabold_1 \| < \delta (\Sigma  )/2$. Suppose that there also exist $\ybold_2 \in \B$ and $\etabold_2$ such that $\rbold = \ybold_2 + \etabold_2$ with $\|\etabold_2 \| < \delta (\Sigma  )/2$. Clearly $\|\etabold_1 - \etabold_2\| < \delta (\Sigma  )$, so that $\ybold_1 - \ybold_2 = \etabold_1 - \etabold_2$ is a trajectory in $\B$ of weight $< \delta (\Sigma  )$. Definition~\ref{def_security} now implies that 
		$\ybold_1 - \ybold_2$ is the zero trajectory, in other words, $\ybold_1 =\ybold_2$ is unique.  It follows that $\|\rbold - \ybold_1\| = \|\etabold_1 \| < \delta (\Sigma  )/2$ is minimal.
\end{proof}

We have formulated system $\Sigma$'s security index $\delta(\Sigma)$ and results in therms of the output trajectory $\ybold$, instead of $\Sigma$'s initial condition $x(0)$ to transparently draw an analogy with error control coding. In fact, this choice is natural because the recovery of $x(0)$ is equivalent to having $\ybold \in \B$ during attack correction, due to the observability assumption on system $\Sigma$.

\section{Computing a system's security index}
In this section, we show how system $\Sigma$'s security index $\delta(\Sigma)$ can be computed. To this end, we introduce the {\em coding matrix} $G$ of the system $\Sigma$ defined as follows 
\beq
G = \bmat{c} G_1 \\ G_2 \\ \vdots \\ G_N \emat,\;\;\;\mbox{where }G_i := \bmat{c} C_i \\
C_iA
\\
C_iA^2 \\
\vdots \\
C_iA^{n-1}
\emat\label{eq_Gi}
\eeq
with $C_i$ defined as the $i$'th row of $C$ for $i=1, \ldots , N$.

We call the above matrix $G$ the {\em coding matrix} of the system as it exhibits the link with error control coding \cite{leeSE2015ECC}. Note however in contrast to error control coding, the coding matrix $G$ cannot be chosen freely---instead, it is fixed and given by the system $\Sigma$. In particular, the number of sensors $N$ is fixed.  In the following theorem, we use $G_\J $ to denote the matrix that is obtained by stacking the matrices $G_i$ defined in~(\ref{eq_Gi}), for $i\in \J \subseteq \{1,\ldots , N\}$.

\begin{thm}\label{thm_GJ}
\[
\delta (\Sigma ) = N-L, 
\]
where $L$ is the largest integer in $\{0,1,\ldots , N\}$ for which there exists a subset $\J$ of $\{1,\ldots , N\}$ of cardinality $L$ such that $\ker G_\J \neq \{ 0 \}$.
\end{thm}
\begin{proof}
		Let $\J$ be a subset of $\{0,1,\ldots , N\}$ of cardinality $L$ such that $\ker G_\J \neq \{ 0 \}$. Let $x(0) \neq 0$ be such that $G_\J  x(0) =0$. Then $Gx(0) \neq 0$  because of left invertibility of $G$. Furthermore, the trajectory $\ybold$ that corresponds to initial condition $x(0)$ satisfies $\ybold_\J = 0$ (use Cayley-Hamilton). Thus $\ybold$ is a trajectory in $\B$ of weight $\leq N-L$, so that $\delta (\Sigma ) \leq N-L$. Further, by definition of $L$, trajectories in $\B$ cannot have more than $L$ of their component trajectories equal to the zero trajectory. Therefore $\delta (\Sigma ) = N-L$.
\end{proof}
\begin{cor} \label{cor:max_secure_G}
The system $\Sigma$ given by equations~(\ref{eq_Asystem})-(\ref{eq_Csystem}) is maximally secure if and only if each matrix $G_i$, as defined in~(\ref{eq_Gi}), has full column rank ($i=1,2,\ldots , N$).
\end{cor}

Not surprisingly, a system is maximally secure if and only if the system is observable via each sensor. In this case, we obtain from Theorem \ref{thm_correct} that the sufficient condition for the attack on the system to be correctable is that the number of compromised sensors is strictly less than half of the total number of sensors. This conforms with the results in \cite{fawziTD14} and \cite{ChongWakaikiHespanhaACC15} for discrete-time and continuous-time LTI systems, respectively.

We further provide ways of computing a system's security index $\delta(\Sigma)$. Since we assume that the LTI system is observable and hence the $nN \times n$ matrix $G$ given by~(\ref{eq_Gi}) has full rank, there exists a full rank $(n(N - 1) \times nN)$ matrix $H$, written as 
\beq
H = \bmat{cccc} H_1 & H_2 & \cdots & H_N \emat ,\label{eq_Hi}
\eeq
such that $HG = 0$; mindful of the analogous coding theoretic terminology, in this paper we call such a matrix $H$ a {\em check matrix} of the system. In the next theorem $H_\J $ denotes the matrix that is obtained by juxtaposing the matrices $H_i$ defined in~(\ref{eq_Hi}), for $i\in \J \subset \{1,\ldots , N\}$.
\begin{thm}\label{thm_HJ}
\[
\delta (\Sigma ) = \spark(H), 
\]
where $\spark(H)$ is defined as the smallest integer $L$ in $\{1,\ldots , N\}$ for which there exists a subset $\J$ of $\{1,\ldots , N\}$ of cardinality $L$ such that $\ker H_\J \neq \{ 0 \}$. 
\end{thm} 
\begin{proof}
		Let $\J$ be a subset of $\{1,\ldots , N\}$ of cardinality $L$ such that $\ker H_\J \neq \{ 0 \}$. Let $\ybold_\J$ be a non-zero trajectory such that $H_\J \ybold_\J =0$. Let $\ybold $ be the trajectory that coincides with $\ybold_\J$ at the appropriate locations and that has zero component trajectories at all other locations. Then $\ybold$ is a trajectory in $\B$ of weight $\leq \spark(H)$, so that $\delta (\Sigma ) \leq \spark(H)$. Further, by definition of $\spark(H)$, trajectories in $\B$ have weight $\geq \spark(H)$ (use Cayley-Hamilton) and therefore $\delta (\Sigma ) = \spark(H) $. 
\end{proof}

The terminology "spark" stems from the compressed sensing literature~\cite{TillmannP14}.

\begin{cor}
The system $\Sigma$ given by equations~(\ref{eq_Asystem})-(\ref{eq_Csystem}) is maximally secure if and only if all square $n(N - 1) \times n(N - 1)$ submatrices of $H$ of the form $H_\J $, where $\J$ is a subset of $\{1,\ldots , N\}$ of cardinality $N-1$, are nonsingular.
\end{cor}

An alternative representation of the system $\Sigma$ is given by a set of $N$ difference equations
\beq
R(\sigma ) \ybold = 0 ,\label{eq_Rkernel}
\eeq
where $R(\xi)$ is a $N\times N$ polynomial matrix and $\sigma$ represents the left shift, as before.
In the special case where $R(\param)$ corresponds to a minimal lag representation, its $N$ row degrees are the observability indices of the system $\Sigma$. 

Recall that a square polynomial matrix is called {\em unimodular} if it has a polynomial inverse; a nonsquare polynomial matrix is called {\em left unimodular} if it has a polynomial left inverse. Two polynomial matrices $R(\param)$ and $Q(\param)$ of the same size are called {\em left unimodularly equivalent} if there exists a unimodular matrix $U(\param)$ such that $Q(\param) = U(\param) R(\param)$.

\begin{thm}\label{thm_RJ}
Let the system $\Sigma$ be given by~(\ref{eq_Rkernel}). Then its security index $\delta (\Sigma ) $ is given by the smallest integer $L$ in $\{1,\ldots , N\}$ for which there exists a subset $\J$ of $\{1,\ldots , N\}$ of cardinality $L$ such that $R_\J(\xi)$ is not left unimodular (here $R_\J(\xi)$ denotes the matrix that consists of the $i$th columns of $R(\xi)$ where $i\in \J$).
\end{thm}
\begin{proof}
		Let $\J$ be a subset of $\{1,\ldots , N\}$ of cardinality $L$ such that $R_\J(\xi)$ is not left unimodular. Then there exists $\ybold_{\J} \neq 0$ such that $R_\J \ybold_{\J} =0$. Let $\ybold $ be the trajectory that coincides with $\ybold_\J$ at the appropriate locations and that has zero component trajectories at all other locations. Then $\ybold$ is a trajectory in $\B$ of weight $\leq L$, so that $\delta (\Sigma ) \leq L$. By definition of $L$ we must also have $\delta (\Sigma ) \geq L$ which proves the theorem.
\end{proof}

In the remainder of this section, we show that the system $\Sigma$'s security index $\delta(\Sigma)$ can be computed more easily by exploiting the system's special structure. Below, $\supp(y)$ denotes the support of a vector $y \in {\mathbb C}^N$, i.e. the set of indices $i$ in $\{ 0,1, 2, \ldots , N\}$ such that $y_i \neq 0$. As mentioned in Definition \ref{def:supp_y_trajectory}, $\supp(\ybold)$ denotes the support of a trajectory $\ybold$, i.e. the set of indices $i$ in $\{ 0,1, 2, \ldots , N\}$ such that $\ybold_i $ is not the zero trajectory. 

\begin{lem}\label{lemma_eigen}
Let the system $\Sigma$ be given by equations~(\ref{eq_Asystem})-(\ref{eq_Csystem}). Let $x(0)$ be a linear combination of eigenvectors of $A$ that correspond to different eigenvalues, so
\beq
x(0) = \alpha_1 v_1 + \alpha_2 v_2 + \cdots + \alpha_m v_m ,\label{eq_alphas}
\eeq
where $0 \neq \alpha_j \in {\mathbb C}$, $Av_j = \lambda_j v_j$ for $j=1,2,\ldots ,m$ and $\lambda_j\neq \lambda_i$ for $j\neq i$.
Let $\ybold$ be the trajectory in the behavior $\B$ of $\Sigma$ that corresponds to $x(0)$. Then
\beq
\supp ( \ybold ) = \cup_{j=1}^m \supp (Cv_j ) . \label{resultSupp}
\eeq
\end{lem}
\begin{proof}
It follows immediately from the definition of $\ybold$ that 
\beq
\supp ( \ybold ) \subset \cup_{j=1}^m \supp (Cv_j ) . \label{pfSupp1}
\eeq
To prove the reverse inclusion, let $\ell \in \cup_{j=1}^m \supp (Cv_j ) $. Then there exists $i \in \{1,2,\ldots , m\}$ such that $C_\ell v_i \neq 0$. Since $\alpha_i$ is assumed non-zero it follows that $\alpha_i C_\ell v_i \neq 0$. Then
\[
\bmat{c} y_\ell(0) \\ y_\ell(1) \\y_\ell(2) \\ \vdots \emat = \bmat{cccc}
1& 1 & \cdots & 1 \\
\lambda_1 & \lambda_2 & \cdots & \lambda_m \\
\lambda_1^2 & \lambda_2^2 & \cdots & \lambda_m^2 \\
\vdots & \vdots & \vdots & \vdots  \emat 
\bmat{c}
\alpha_1 C_\ell v_1 \\
\alpha_2 C_\ell v_2 \\
\vdots \\
\alpha_m C_\ell v_m \emat 
\]
is non-zero because of the Vandermonde structure and the fact that all $\lambda_i$'s  are distinct. Thus $\ybold_\ell$ is not the zero trajectory so that $\ell \in \supp(\ybold )$.
This implies that
\beq
\cup_{j=1}^m \supp (Cv_j ) \subset \supp ( \ybold ) . \label{pfSupp2}
\eeq
From~(\ref{pfSupp1}) and~(\ref{pfSupp2}) we conclude that~(\ref{resultSupp}) holds. 	
\end{proof}
A special consequence of the above lemma is that the system's security index $\delta(\Sigma)$ is determined by the weight of the special trajectories that have initial conditions in an eigenspace $V_i$ of $A$: 
\begin{thm}\label{thm_eigen}
Let the system $\Sigma$ be given by equations~(\ref{eq_Asystem})-(\ref{eq_Csystem}). Let $\lambda_1, \lambda_2 , \ldots , \lambda_m$ be the $m$ distinct eigenvalues of the matrix $A$. Let ${\mathbb C}^n = V_1 \cup  V_2 \cup \ldots \cup V_m$, where $V_i$ is the subspace spanned by all eigenvectors corresponding to $\lambda_i$ for $i=1,2,\ldots ,m$. 
Then
\[
\delta (\Sigma ) = \min_{j \in \{1,2,\ldots , m\}} \min_{x\in V_j} |\supp(Cx) |,
\]
In particular, if $m=n$, meaning that all $n$ eigenvalues of $A$ are distinct, then
\[
\delta (\Sigma ) =  \min_{j \in \{1,2,\ldots , n\}} |\supp(Cv_j)| ,
\]
where $v_1 , v_2 , \ldots , v_n$ is a basis of eigenvectors of $A$.
\end{thm}
\begin{proof}
Let $\ybold$ be an arbitrary trajectory in $\B$, corresponding to initial condition $x(0)$, written as in~(\ref{eq_alphas}). It follows from Lemma~\ref{lemma_eigen} that 
\[
\|\ybold \|= |\supp ( \ybold ) | \geq \min_{j \in \{1,2,\ldots , m\}} |\supp (Cv_j )| . 
\]
This implies that
\[
\delta (\Sigma ) \geq \min_{j \in \{1,2,\ldots , m\}} \min_{x\in V_j} |\supp(Cx) |.
\]
		To prove the reverse inequality, let $j \in \{1,2,\ldots , m\}$ and $x \in V_j$ be such that $|\supp(Cx)|$ is minimal. Let $\ybold \in \B$ be the system's trajectory that corresponds to initial condition $x(0) := x$. From $Ax = \lambda_j x$ it follows immediately that $\|\ybold \|= |\supp(Cx)|$. Therefore 
		\[
		\delta (\Sigma ) \leq \min_{j \in \{1,2,\ldots , m\}} \min_{x\in V_j} |\supp(Cx) | . 
		\]
		which proves the theorem.
\end{proof}
From the theorem above, we see that the computation of the security index is straightforward for a system that has $n$ distinct eigenvalues : simply use a diagonal $A$-matrix; the security index of the system is then given by the size of the support of the sparsest column of the corresponding $C$-matrix.  In particular, the system is maximally secure if and only if this $C$-matrix has no zero values. We now illustrate the computation of system $\Sigma$'s security index through an example. 
\begin{example}
Let $N=2$ and consider
\[
A= \bmat{cc} \lambda_1 & 0 \\ 0 & \lambda_2 \emat, \;\;\;\;C = \bmat{cc} 1 & 1 \\ 1 & -1 \emat .
\]
Then it follows from Theorem~\ref{thm_eigen} that for $\lambda_1 \neq \lambda_2$ the system's security index is $2$ so that by Theorem~\ref{thm_detect} detection of attacks on one output is possible. Minimal lag equations for this system are given by
\begin{equation}
	R(\sigma) \ybold = 
\bmat{cc} \sigma - \lambda_1 & \sigma - \lambda_1  \\  \sigma - \lambda_2 & -(\sigma - \lambda_2 ) \emat \bmat{c} \ybold_1 \\ \ybold_2 \emat = 0. \nonumber
\end{equation}
We observe that when $\lambda_{1}\neq \lambda_2$, every column of $R(\sigma)$ is left unimodular. Hence, by Theorem \ref{thm_RJ}, we obtain $L=N$ and thus, the security index $\delta(\Sigma)$ equals $N$.
\end{example}

\section{Attack detection, correction and other future work}
Using the results developed in the previous sections, we now discuss how they can be employed in detecting the occurrence of an attack and how correction can be performed.

To start with detection, we first recall that the observability assumption on our system implies that a check matrix $H$ (defined in \eqref{eq_Hi}) exists. Attack detection is most easily formulated in terms of $H$ as follows.

\begin{quote} 
{\em Attack detection rule:} Given received signal $\rbold \in \B'$, compute the `syndrome trajectory' $\sbold := H \rbold $ and conclude that an attack has taken place (meaning $\rbold \notin \B$) if and only if $\sbold \neq 0$. 
\end{quote}

Alternatively, we can use the description of the form $R(\sigma ) \ybold =0$ as the main ingredient of our decision rule to detect/correct.   
\begin{quote}
{\em Attack detection rule:} Given received signal $\rbold \in \B'$, compute the `syndrome trajectory' $\sbold := R(\sigma )\rbold $ and conclude that an attack has taken place (meaning $\rbold \notin \B$) if and only if $\sbold \neq 0$. 
\end{quote}
In this paper, we choose to formulate all fundamental notions and results in terms of the trajectory $\ybold$ rather than in terms of the initial condition $x(0)$ so as to have a clear fundamental theory that exhibits the link with error control coding in a transparent way.  Indeed, we are able to make this choice because for attack correction the recovery of $x(0)$ is equivalent to the recovery of $\ybold \in \B$ due to the observability assumption on the system $\Sigma$. It is however important to note that in practical situations the recovery of $x(0)$ from the received trajectory $\rbold$ is the main objective. In fact, we only need to consider attack scenarios where attacks happen in the first $n$ time instances (otherwise we can simply reconstruct $x(0)$ from $(r(0) , r(1), \cdots , r(n-1))=(y(0) , y(1), \cdots , y(n))$ because of the observability of $\Sigma$).

For attack correction, the above syndrome trajectory $\sbold$ should be used to identify the attack locations.  Once the attack locations have been found, $x(0)$ can be computed on the basis of the remaining attack-free components of $\ybold$. How to identify the attack locations from $\sbold$ in terms of our knowledge of the system matrices $A$ and $C$ is a topic of future research. Another topic of future research is the effect of feedback control on the security index of an LTI system with inputs.

\addtolength{\textheight}{-12cm}   





\bibliographystyle{plain}
\bibliography{code}

\end{document}